\documentclass[11pt]{amsart}
\usepackage{mathrsfs}
\usepackage{amssymb} 
\usepackage{url}
\usepackage{hyperref}
\usepackage{color}
\definecolor{refcolor}{RGB}{0,0,190}
\hypersetup{
    colorlinks,
    citecolor=refcolor,
    filecolor=refcolor,
    linkcolor=refcolor,
    urlcolor=refcolor
}
\usepackage{array}
\usepackage{booktabs}
\theoremstyle{definition}
\newtheorem{theorem}{Theorem}[section]
\newtheorem{definition}[theorem]{Definition}

\def\({\left(}
\def\){\right)}

\newcommand{\R}{\mathbb{R}}

\newcommand{\de}{\textnormal{d}}

\newcommand{\tn}{\textnormal}
\newcommand{\ds}{\displaystyle}

\newcommand{\ie}{\textit{i.e.} }

\newcommand{\eg}{\textit{e.g.} }

\newcommand{\citep}[2]{\cite{#1}, p. #2}

\newcommand{\Ric}{\textnormal{Ric}}

\newcommand{\mc}[1]{\mathcal{#1}}

\newcommand{\sref}[1]{\S\ref{#1}}

\newcommand{\dsfrac}[2]{\ds{\frac{#1}{#2}}}

\newcommand{\metric}[1]{\langle#1\rangle}
\newcommand{\kosz}{\mc K}

\newcommand{\cocontr}{{{}_\bullet}}

\newcommand{\idxannih}[2]{#1{}^{#2}{}}
\newcommand{\idxcoannih}[2]{#1{}_{#2}{}}

\newcommand{\radix}[1]{\idxcoannih{#1}{\circ}}
\newcommand{\annih}[1]{\idxannih{#1}{\bullet}}
\newcommand{\coannih}[1]{\idxcoannih{#1}{\bullet}}

\newcommand{\annihg}{\coannih{g}}

\newcommand{\IM}{\tn{im }}

\def\hyph{-\penalty0\hskip0pt\relax}
\newcommand{\semiriem}{semi{\hyph}Riemannian}

\newcommand{\semireg}{semi{\hyph}regular}
\newcommand{\ssemireg}{Semi{\hyph}regular}
\newcommand{\quasireg}{quasi{\hyph}regular}

\newcommand{\nondeg}{non{\hyph}degenerate}

\newcommand{\flrw}{Friedmann-Lema\^itre-Robertson-Walker}
\newcommand{\FLRW}{FLRW}
\newcommand{\schw}{Schwarzschild}

\newcommand{\Wch}{Weyl curvature hypothesis}
\newcommand{\WCH}{WCH}

\begin{document} 
 
\title{On the Weyl Curvature Hypothesis\footnote{Annals of Physics}}
\author{Ovidiu Cristinel \ Stoica}
\date{July 5, 2013}

\begin{abstract}
The Weyl curvature hypothesis of Penrose attempts to explain the high homogeneity and isotropy, and the very low entropy of the early universe, by conjecturing the vanishing of the Weyl tensor at the Big-Bang singularity.

In previous papers it has been proposed an equivalent form of Einstein's equation, which extends it and remains valid at an important class of singularities (including in particular the Schwarzschild, FLRW, and isotropic singularities). Here it is shown that if the Big-Bang singularity is from this class, it also satisfies the Weyl curvature hypothesis.

As an application, we study a very general example of cosmological models, which generalizes the FLRW model by dropping the isotropy and homogeneity constraints. This model also generalizes isotropic singularities, and a class of singularities occurring in Bianchi cosmologies. We show that the Big-Bang singularity of this model is of the type under consideration, and satisfies therefore the Weyl curvature hypothesis.

\bigskip
\noindent 
\end{abstract}


\maketitle

\setcounter{tocdepth}{2}
\tableofcontents

\section{Introduction}
\label{s_intro}

\subsection{The \Wch}
\label{s_wch}

In his seminal paper \cite{Pen79}, R. Penrose was interested in several distinct problems, including the search for an explanation of the second law of Thermodynamics, and of the high homogeneity and isotropy of the universe. His analysis of the flow of energy in the Universe led him to the idea that the second law of Thermodynamics is due to very high homogeneity around the Big-Bang. Penrose explains this homogeneity by the following argument (\cite{Pen79}, p. 614):

\begin{quote}
In terms of spacetime curvature, the absence of clumping corresponds, very roughly, to the absence of Weyl conformal curvature (since absence of clumping implies spatial-isotropy, and hence no gravitational principal null-directions).
\end{quote}

He further added ``this restriction on the early geometry should be something like: the Weyl curvature $C_{abcd}$ vanishes at any initial singularity'' (\cite{Pen79}, p. 630). This is the \textit{\Wch} (\WCH), and it makes the central point of interest in the present article.

In addition to Penrose's motivations for the {\WCH}, other reasons come from Quantum Gravity. We expect that, near the Big-Bang, the quantum effects of gravity take over, but we know that gravity is perturbatively nonrenormalizable at two loops \cite{HV74qg,GS86uvgr}. A vanishing Weyl tensor would mean a vanishing of local degrees of freedom, hence of gravitons, and this would remove some of the problems  of Quantum Gravity, at least at the Big-Bang \cite{Car95}.

The Weyl tensor $C_{abcd}$ is the traceless part of the Riemann curvature tensor $R_{abcd}$. From gravitational viewpoint, it is responsible for the tidal forces. The tensor $C_{abc}{}^d$ is invariant at the conformal rescaling $g_{ab}\mapsto\Omega^2 g_{ab}$. If it vanishes, it indicates that the metric is conformally flat (in dimension $\geq 4$; in lower dimension it vanishes trivially).

An important class of singularities which automatically satisfy a version of the {\WCH} (stating that the Weyl curvature remains finite at singularity) are the \textit{isotropic singularities}. They were researched by Tod \cite{Tod87,Tod90,Tod91,Tod92,Tod02,Tod03}, Claudel \& Newman \cite{CN98}, Anguige \& Tod \cite{AT99i,AT99ii}. The metric tensor in this case can be obtained by a conformal rescaling from a regular (\ie non-singular) metric tensor, and presents nice behavior from conformal geometric viewpoint. Because the Weyl tensor $C_{abc}{}^d$ is invariant at conformal rescaling, the main feature of the isotropic singularities is that the Weyl tensor equals that of a regular metric, hence remaining finite. If we apply a conformal rescaling to a metric tensor $g_{ab}$, the new metric $\hat{g}_{ab}:=\Omega^2g_{ab}$ has the Weyl curvature tensor $\hat{C}_{abc}{}^d = C_{abc}{}^d$, which is smooth, but not necessarily vanishing at $\Omega=0$. But $\hat{C}_{abcd}$ vanishes, since
\begin{equation}
	\hat{C}_{abcd}=\hat{g}_{sd}\hat{C}_{abc}{}^s=\Omega^2 g_{sd}C_{abc}{}^s=0.
\end{equation}

A simple example of vanishing Weyl tensor is provided by the {\flrw} ({\FLRW}) cosmological model. This model implements the \textit{cosmological principle} that the Universe is, at very large scales, homogeneous and isotropic. But the fact that the {\FLRW} model has vanishing Weyl curvature is irrelevant for the {\WCH}, because it is too symmetric, and its Weyl tensor vanishes identically.

In this paper we will prove that a much larger class of singularities, which includes among others the isotropic singularities and the {\FLRW} ones, satisfies the {\Wch}. This larger class is not a random generalization of the isotropic singularities, but it appeared from our research on a different problem. Previous considerations led us to the idea that, for a singularity to be still manageable from mathematical viewpoint, some conditions are in order \cite{Sto11a,Sto11b,Sto12e}. These conditions allowed the definition of a smooth Riemann curvature $R_{abcd}$ (but not of $R_{abc}{}^d$) at singularities. This program successfully led to a better understanding of the black hole singularities (see \eg \cite{Sto11a,Sto11b,Sto12e,Sto11e,Sto11f}).

Further considerations suggested that the Ricci decomposition of the curvature tensor should be smooth, in order to allow the writing of an equation which extends Einstein's at a large class of singularities, but is equivalent to it for {\nondeg} metrics \cite{Sto12a}. These singularities behave, mathematically, as if the metric loses one or more dimensions, and consequently as if the cotangent space loses one or more dimensions. They have, as a bonus, the property that $C_{abcd}$ vanishes, as we will prove in \sref{s_wch_thm}. The proof is largely based on the idea that $C_{abcd}=0$ in lower dimensions.

As an application, in \sref{s_wch_ex} we will study a very general cosmological model, which does not assume, like the {\FLRW} model, that the space slices have the metric constant in time, up to an overall scaling factor $a^2(t)$ which depends on the time only and which vanishes at the Big-Bang. In our model we will keep the overall scaling factor $a^2(t)$, but we will allow the space part of the metric to change in time freely. The physical motivation is that in the actual Universe there is no perfect isotropy and homogeneity, and in fact at small scales the inhomogeneity is important. Working in such general settings, we will not be concerned at this time with the matter content of this universe. Because the solution is very general, the possibilities of matter fields which can give this kind of metric are limitless. This metric has non-vanishing Weyl tensor in general, but we will show that at the Big-Bang singularities $C_{abcd}$ vanishes.

\subsection{{\ssemireg} spacetimes}
\label{s_sreg}

We were initially interested in the general study of metrics which can become degenerate or change their signature. We arrived at a large class of metrics with singularities, for which the Riemann curvature tensor $R_{abcd}$ can be defined and remains smooth. The main obstruction is, of course, the fact that when the metric $g_{ab}\to 0$, its reciprocal $g^{ab}\to\infty$. This prevents standard tensor operations like index raising, and the contraction between covariant indices. Consequently, we cannot define the Levi-Civita connection and its curvature. But we found a large class of singular {\semiriem} manifolds on which we can define covariant contraction, and eventually construct directly the Riemann curvature tensor $R_{abcd}$, which turns out to be smooth (for this section, please refer to \cite{Sto11a,Sto11b,Sto12e}). To do this, we avoided working with $\nabla_XY$, which cannot be defined, and used instead as much as possible the \textit{Koszul object}, defined (see \eg \citep{Kup87b}{263}, and \cite{Sto11a,Sto11b}) as the right side of the \textit{Koszul formula} (see \eg \citep{ONe83}{61}):
\begin{equation}
\label{eq_Koszul_form}
	\kosz(X,Y,Z) :=\ds{\frac 1 2} \{ X \metric{Y,Z} + Y \metric{Z,X} - Z \metric{X,Y} 
	- \metric{X,[Y,Z]} + \metric{Y, [Z,X]} + \metric{Z, [X,Y]}\}
\end{equation}
which helps, for a {\nondeg} metric, defining the Levi-Civita connection implicitly by
\begin{equation}
	\metric{\nabla_XY,Z}=\kosz(X,Y,Z).
\end{equation}
When the metric becomes degenerate we cannot construct the Levi-Civita connection.

We define the Riemann curvature $R_{abcd}$ even for {\nondeg} metrics, in a way which avoids the undefined $\nabla_XY$, but relies on the defined and smooth $\kosz(X,Y,Z)$.

To be able to define the Riemann curvature in terms of the Koszul object only, the differential $1$-form $\kosz(X,Y,\_)$ has to admit covariant derivative, which has to be smooth. This is equivalent to requiring that
\begin{enumerate}
	\item
	$\kosz(X,Y,W)=0$ whenever $W$ satisfies $\metric{W,X}=0$ for any $X$
	\item
	the contraction $\kosz(X,Y,\cocontr)\kosz(Z,T,\cocontr)$ is smooth for any local vector fields $X,Y,Z,T$.
\end{enumerate}

A \textit{{\semireg} manifold} is defined as a singular {\semiriem} manifold satisfying the above conditions. Its metric is called \textit{{\semireg} metric}.
A \textit{{\semireg} spacetime} is a $4$-dimensional {\semireg} manifold $M$ with metric $g$ which 
\begin{enumerate}
	\item 
has the signature at each point $(r,s,t)$, $s\leq 3$, $t\leq 1$ (in other words, the metric can be diagonalized at each point to a matrix of the form $\tn{diag}(\underbrace{0,\ldots,0}_{r\tn{ times}},\underbrace{1,\ldots,1}_{s\tn{ times}},\underbrace{-1,\ldots,-1}_{t\tn{ times}})$). 
	\item 
	and is {\nondeg} on a dense subset of $M$. In fact, for the purpose of black hole and big-bang singularities, this condition is too general. The known singularities are submanifolds of dimension $<4$, so we can safely replace this condition with the stronger condition that the metric is regular on the entire manifold, with the exception of a finite (or at most countable) union of lower-dimensional manifolds.
\end{enumerate}

We defined the Riemann curvature $R_{abcd}$ by
\begin{equation}
\label{eq_riemann_curvature_tensor_coord}
	R_{abcd}= \partial_a \Gamma_{bcd} - \partial_b \Gamma_{acd} + \Gamma_{ac\cocontr}\Gamma_{bd\cocontr} - \Gamma_{bc\cocontr}\Gamma_{ad\cocontr}
\end{equation}
where $\Gamma_{abc}=\kosz(\partial_a,\partial_b,\partial_c)$ are Christoffel's symbols of the first kind. This is smooth for {\semireg} manifolds.

It is easy to see that from the smoothness of $R_{abcd}$ follows the smoothness of the Einstein tensor density $G_{ab}\det g$, leading to the possibility of writing a densitized version of Einstein's equation, which remains smooth at {\semireg} singularities \cite{Sto11a}.

But there is another way to write an equivalent form of Einstein's equation, which we introduced in \cite{Sto12a}. It leads to an \textit{expanded Einstein equation}, obtained from Einstein's equation by taking the Kulkarni-Nomizu product (section \sref{s_einstein_exp_qreg}, eqn. \eqref{eq_kulkarni_nomizu}). This leads to a special class of {\semireg} spacetimes, named {\quasireg} spacetimes, which allow the formulation of the expanded Einstein equation. As it turned out, important classes of singularities are {\quasireg}: the isotropic singularities, the {\schw} singularity, $1+3$ degenerate warped products, in particular the {\FLRW} Big-Bang singularity. As we will see, the quasiregular singularities satisfy the {\WCH}.

\subsection{Expanded Einstein equation and {\quasireg} spacetimes}
\label{s_einstein_exp_qreg}

The standard Einstein equation is
\begin{equation}
\label{eq_einstein}
	G_{ab} + \Lambda g_{ab} = \kappa T_{ab}
\end{equation}
where $\Lambda$ is the \textit{cosmological constant}, and $\kappa:=\dsfrac{8\pi \mc G}{c^4}$, $\mc G$ and $c$ being the gravitational constant and the speed of light. 
By
\begin{equation}
\label{eq_einstein_tensor}
	G_{ab}:=R_{ab}-\frac 1 2 R g_{ab},
\end{equation}
we denoted the Einstein tensor, obtained from the \textit{Ricci curvature} $R_{ab} := g^{st}R_{asbt}$ and the \textit{scalar curvature} $R := g^{st}R_{st}$.
Einstein's equation \eqref{eq_einstein} makes sense so long as the metric is not singular.

One situation when the metric becomes singular is by becoming degenerate -- \ie its determinant becomes $0$. Then, the inverse $g^{ab}$ cannot be constructed, and Einstein's equation can no longer be written.

The main ingredient we utilized to smoothen the Einstein equation at such singularities is the \textit{Kulkarni-Nomizu product} of two symmetric bilinear forms $h$ and $k$, defined as
\begin{equation}
\label{eq_kulkarni_nomizu}
	(h\circ k)_{abcd} := h_{ac}k_{bd} - h_{ad}k_{bc} + h_{bd}k_{ac} - h_{bc}k_{ad}.
\end{equation}
With its help, in \cite{Sto12a} we introduced a new version of Einstein's equation, called \textit{the expanded Einstein equation}:
\begin{equation}
\label{eq_einstein_expanded}
	(G\circ g)_{abcd} + \Lambda (g\circ g)_{abcd} = \kappa (T\circ g)_{abcd}.
\end{equation}

The new equation is equivalent to the standard Einstein equation so long as the metric is not singular, but it can be written, and remains smooth, even in cases when the metric becomes degenerate. When the metric becomes degenerate, it vanishes for some directions in spacetime. The inverse of the metric becomes singular in these directions, making the Einstein tensor $G_{ab}$ singular, but the Kulkarni-Nomizu product $(G\circ g)_{abcd}$ contains additional factors $g_{ab}$, which compensate the singularities of the metric's inverse, so that $(G\circ g)_{abcd}$ is smoothened \cite{Sto12a}.

The tensor $S_{abcd}$ is the scalar part of the Riemann curvature, and $E_{abcd}$ the \textit{semi-traceless part} of the Riemann curvature. They are defined (see \eg \cite{ST69,BESS87,GHLF04}), for a {\semiriem} manifold of dimension $n$, by
\begin{equation}
	S_{abcd} = \dsfrac{1}{2n(n-1)}R(g\circ g)_{abcd},
\end{equation}
\begin{equation}
	E_{abcd} = \dsfrac{1}{n-2}(S \circ g)_{abcd},
\end{equation}
where 
\begin{equation}
\label{eq_ricci_traceless}
S_{ab} := R_{ab} - \dsfrac{1}{n}Rg_{ab}.
\end{equation}

The \textit{Weyl curvature tensor} is defined by
\begin{equation}
\label{eq_weyl_curvature}
	C_{abcd} = R_{abcd} - S_{abcd} - E_{abcd}.
\end{equation}

The expanded Einstein equation can be rewritten in terms of $E_{abcd}$ and $S_{abcd}$:
\begin{equation}
\label{eq_einstein_expanded_explicit}
	2 E_{abcd} - 6 S_{abcd} + \Lambda (g\circ g)_{abcd} = \kappa (T\circ g)_{abcd}.
\end{equation}

In \cite{Sto12a} we saw that our version of Einstein's equation remains smooth at the FLRW Big-Bang singularity, isotropic singularities, $1+3$ warped products with degenerate warping function, and even at the {\schw} black hole singularity \cite{Sto11e}.

This motivated the introduction of {\quasireg} spacetimes in \cite{Sto12a}.
\begin{definition}
\label{def_quasi_regular}
We say that a {\semireg} manifold $(M,g_{ab})$ is \textit{{\quasireg}}, and that $g_{ab}$ is a \textit{{\quasireg} metric}, if:
\begin{enumerate}
	\item 
$g_{ab}$ is {\nondeg} on a subset $M_0\subset M$ dense in $M$ (so that $\overline {M_0}=M$)
	\item 
	the tensors $S_{abcd}$ and $E_{abcd}$ defined at the points $p\in M_0$ where the metric is {\nondeg} extend smoothly to the entire manifold $M$.
\end{enumerate}
If a {\semireg} spacetime is also a {\quasireg} manifold, we call it \textit{{\quasireg} spacetime}.
\end{definition}

From mathematical viewpoint, the condition that $\overline {M_0}=M$ is required to allow the extension by continuity of the tensors $S_{abcd}$ and $E_{abcd}$ at the points where the metric is {\nondeg}. From physical viewpoint this condition is not limiting, because we do not expect the metric to be degenerate on four-dimensional volumes. We do not know of examples (even in theory) of a metric which is degenerate on a four-dimensional volume, and in fact, such a degeneracy would extend to the entire spacetime, because of the field equations.

The central result of this paper is that the {\quasireg} singularities satisfy the {\Wch}.

\section{The Weyl curvature vanishes at {\quasireg} singularities}
\label{s_wch_thm}

We will now prove that the Weyl curvature tensor $C_{abcd}$ vanishes at {\quasireg} singularities. Its smoothness entails that $C_{abcd}\to 0$ as approaching such singularities. On the other hand, it is allowed to increase indefinitely away from the singularity. In fact, the condition that the metric is {\quasireg} introduces constraints only at the singularities, since at the points where the metric is regular, it is automatically {\quasireg}. In particular, the Weyl tensor vanishes at a {\quasireg} Big-Bang, but can increase as the universe subsequently evolves, as required by the \Wch.

\begin{theorem}
The Weyl curvature tensor $C_{abcd}$ vanishes at {\quasireg} singularities.
\end{theorem}
\begin{proof}
From the smoothness of $R_{abcd}$, $E_{abcd}$, $S_{abcd}$, and from equation \eqref{eq_weyl_curvature}, follows the smoothness of $C_{abcd}$.

The following considerations use objects described in \cite{Sto11a,Sto11b,Sto12e}, but we try to make the proof as self-contained as possible. Since the metric $g$ is a bilinear form on the tangent vector space $T_pM$, it defines, at the points $p$ where is degenerate, the totally degenerate space $\radix{T}_pM:=T_pM^\perp$, named the \textit{radical} of $T_pM$. We construct its annihilator
\begin{equation}
	\annih{T}_pM:=\left\{\omega\in T^*_pM;\omega|_{\radix{T}_pM}=0\right\},
\end{equation}
named the \textit{radical annihilator}. The radical annihilator is the image of the index lowering morphism $\flat:T_pM\to T^*_pM$, $X^\flat(Y):=\metric{X,Y}$, $\forall X,Y\in T_pM$:
\begin{equation}
\annih{T}_pM=\IM\flat\leq T^*_pM.
\end{equation}
The dual of the radical annihilator, $(\annih{T}_pM)^*$, is isomorphic with the quotient $T_pM/\radix{T}_pM$. On $\annih{T}_pM$, $g$ induces a canonical {\nondeg} metric defined by
\begin{equation}
\annihg(\omega,\tau):=\metric{X,Y}
\end{equation}
where $X^\flat=\omega$ and $Y^\flat=\tau$. This metric is used for covariant contractions, so long as the contracted tensor components live in $\annih{T}_pM$.

As it is shown in \cite{Sto11a}, the Riemann curvature tensor $R_{abcd}$ of a {\semireg} {\semiriem} manifold satisfies at each $p\in M$
\begin{equation}
\left(R_{abcd}\right)_p \in \otimes^4 \annih{T}_pM.
\end{equation}
At the points $p$ where the metric $g$ is degenerate, $\dim\left(\annih{T}_pM\right)<4$.
Since in dimension $\leq 3$ any tensor having the symmetries of the Weyl tensor vanishes (see \eg \cite{BESS87}), it follows that 
\begin{equation}
\left(C_{abcd}\right)_p=0
\end{equation}
whenever $g$ is degenerate at $p$. This concludes the proof.
\end{proof}

\section{Example: a general, non-isotropic and inhomogeneous cosmological model}
\label{s_wch_ex}

The {\FLRW} model is very good for very large scales, where we can completely ignore any inhomogeneity and anisotropy in the distribution of matter. But in reality the universe is not homogeneous and isotropic at all scales. This is why physicists studied more general solutions, like those obtained by dropping the condition of isotropy and maintaining homogeneity, such as various Bianchi cosmological models \cite{plebanski2006grc,stephani2003exactsolutions,ellis1999cosmological}. Here we will explicitly construct a very general spacetime $(M,g)$, which is allowed to be inhomogeneous and anisotropic.

Penrose motivated {\WCH} by the necessity to explain the high homogeneity and isotropy of the early universe. Hence, the importance of our example is that it exhibits vanishing Weyl curvature at the singularity, without imposing homogeneity and isotropy. Our choice is also motivated by the fact that it includes as particular cases {\FLRW} spacetimes, isotropic Big-Bang solutions, and some singularities occurring in Bianchi cosmologies, but it is much more general. In this section, we will show that this explicit example of metric is {\quasireg}.

We assume that $M$ has the form $M=I\times \Sigma$, where $I\subseteq \R$ is an interval, and $\Sigma$ a three-dimensional manifold. We consider a global time coordinate $\tau:M\to I$, defined by $\tau(t,x)=t$. We consider on each slice $\Sigma_t=\tau^{-1}(t)$ a Riemannian (non-degenerate) metric $h_{ij}(t,x)$ (where $1\leq i,j\leq 3$), which depends smoothly on $(t,x)\in I\times\Sigma$. We also represent the metric $h(t)$ as an arc element by $\de\sigma_t$. We assume further that the metric $g$ on the manifold $M$ has the form
\begin{equation}
	g_{ij}(t,x):=-\de t\otimes\de t + a^2(t)h_{ij}(t,x),
\end{equation}
where $a:I\to\R$ is a smooth function. While $h(t,x)$ is Riemannian (and {\nondeg}) for any $(t,x)\in I\times\Sigma$, $a(t)$ is allowed to vanish. The Big-Bang singularity is therefore obtained for $a(t)=0$.

On short, we consider that the spacetime is $M=I\times\Sigma$, with the following metric:
\begin{equation}
\label{eq_metric_a_h}
\de s^2 = -\de t^2 + a^2(t)\de\sigma_t^2.
\end{equation}

If we let $h_{ij}$ depend on time, but impose various symmetry conditions, we can obtain various Bianchi cosmologies (see \eg \cite{stephani2003exactsolutions}). If we take $h_{ij}$ to be independent on time, and of constant curvature, we obtain the {\FLRW} model (which is not interesting from the viewpoint of the \WCH, because its Weyl curvature vanishes trivially). But if we allow the geometry of space slices $\(\Sigma_t,h(t)\)$ to be inhomogeneous and to vary with time, the solutions are much more general.

We can actually be a bit more general and allow the metric on $I$ to become degenerate too:
\begin{equation}
\label{eq_metric_a_Nh}
\de s^2 = -N^2(t)\de t^2 + a^2(t)\de\sigma_t^2,
\end{equation}
where $N:I\to\R$ is a smooth function. If $N(t)\neq 0$ for any $t\in I$, then a reparametrization of $I$ allows the metric to be of the form \eqref{eq_metric_a_h}, so this generalization is important only when $N(t)$ vanishes together with $a(t)$. For reasons which will become apparent, we will require that
\begin{equation}
\label{eq_f_aN}
	f(t):=\dsfrac{a(t)}{N(t)}
\end{equation}
is not singular. For example, if $f(t)=1$ (hence $N(t)=a(t)$), the resulting singularities are just isotropic singularities, as those studied by Tod \& \textit{al.}. But allowing $f(t)$ to vanish together with $a(t)$ leads to more general, anisotropic singularities.

We are here interested in the most general case.

\begin{theorem}
\label{thm_metric_a_Nh}
The metric \eqref{eq_metric_a_Nh}, satisfying \eqref{eq_f_aN}, is {\quasireg}.
\end{theorem}
\begin{proof}
The plan is to prove that the metric is {\semireg}, by showing that the terms in the Riemann curvature tensor \eqref{eq_riemann_curvature_tensor_coord} are smooth. Then we show that the Ricci decomposition
\begin{equation}
\label{eq_ricci_decomposition}
	R_{abcd} = E_{abcd} + S_{abcd} + C_{abcd}.
\end{equation}
is smooth.

The metric components are
\begin{equation}
	g(t,x) =
	\left(
\begin{array}{cc}
	-N^2(t) & 0 \\
	0 &  a^2(t) h_{ij}(t,x) \\
\end{array}
\right)
\end{equation}

The reciprocal metric components are
\begin{equation}
	g^{-1}(t,x) =
	\left(
\begin{array}{cc}
	-N^{-2}(t) & 0 \\
	0 &  a^{-2}(t) h^{ij}(t,x) \\
\end{array}
\right)
\end{equation}

The partial derivatives of the metric tensor are therefore
\begin{equation}
\label{eq_g_der}
\begin{array}{ll}
	g_{00,0} &= -2N\dot N \\
	g_{00,k} &= 0 \\
	g_{ij,0} &= a(2\dot a h_{ij} + a \dot h_{ij}) \\
	g_{ij,k} &= a^2 \partial_k h_{ij} \\
\end{array}
\end{equation}

The second order partial derivatives of the metric tensor are

\begin{equation}
\label{eq_g_der_der}
\begin{array}{ll}
	g_{00,00} &= -2\(\dot N^2 + N\ddot N\) \\
	g_{00,k0} &= g_{00,0k} = g_{00,kl} =  0 \\
	g_{ij,00} &= 2\dot a^2 h_{ij} + 2a\ddot a h_{ij}  + 4a\dot a \dot h_{ij} + a^2 \ddot h_{ij} \\
	g_{ij,k0} &= a\(2 \dot a \partial_k h_{ij} + a \partial_k \dot h_{ij}\) \\
	g_{ij,kl} &= a^2 \partial_k \partial_l h_{ij}\\
\end{array}
\end{equation}

To check that $g$ is {\semireg}, it is enough to check that the terms of the form $g_{ab,\cocontr}g_{cd,\cocontr}$ are smooth. By using the equations \eqref{eq_g_der} we find that

\begin{equation}
\begin{array}{lll}
	g_{00,\cocontr}g_{00,\cocontr} &=& -N^{-2}g_{00,0}g_{00,0} + a^{-2}h^{cd}g_{00,c}g_{00,d} \\
	&=& -4\dot N^2 \\
\end{array}
\end{equation}

\begin{equation}
\begin{array}{lll}
	g_{00,\cocontr}g_{ij,\cocontr} &=& -N^{-2}g_{00,0}g_{ij,0} + a^{-2}h^{cd}g_{00,c}g_{ij,d} \\
	&=& 2\dsfrac{\dot N a}{N}\(2\dot a h_{ij} + a\dot h_{ij}\)  \\
	\end{array}
\end{equation}

\begin{equation}
\begin{array}{lll}
	g_{ij,\cocontr}g_{kl,\cocontr} &=& -N^{-2}g_{ij,0}g_{kl,0} + a^{-2}h^{cd}g_{ij,c}g_{kl,d} \\
	&=& -\dsfrac{a^2}{N^2}(2\dot a h_{ij} + a \dot h_{ij})(2\dot a h_{kl} + a \dot h_{kl})
	+ a^2h^{cd}\partial_c h_{ij}\partial_d h_{kl} \\
	\end{array}
\end{equation}

But since for a smooth function $f(t,x)$
\begin{equation}
a(t,x) = f(t,x)N(t),
\end{equation}
the terms calculated above become now manifestly smooth:
\begin{equation}
\label{eq_g_cocontr}
\begin{array}{ll}
	g_{00,\cocontr}g_{00,\cocontr} &= -4\dot N^2 \\
	g_{00,\cocontr}g_{ij,\cocontr} &= 2\dot N f \(2\dot a h_{ij} + a\dot h_{ij}\)  \\
	g_{ij,\cocontr}g_{kl,\cocontr} &= f^2\( - (2\dot a h_{ij} + a \dot h_{ij})(2\dot a h_{kl} + a \dot h_{kl})
		+ N^2h^{cd}\partial_c h_{ij}\partial_d h_{kl}\)\\
\end{array}
\end{equation}
Therefore the metric $g$ is {\semireg}.

We are now interested in proving that $\Ric\circ g$ and $R g \circ g$ are smooth. For this, we have to contract the terms from \eqref{eq_g_der_der} and \eqref{eq_g_cocontr}, and see what happens when taking Kulkarni-Nomizu products with $g$. The tensor $\Ric\circ g$ is a sum of products between $g_{ab,cd}$ or $g_{ab,\cocontr}g_{cd,\cocontr}$, and $g^{ef}$, and $g_{gh}$. The tensor $R g \circ g$ is a sum of products between $g_{ab,cd}$ or $g_{ab,\cocontr}g_{cd,\cocontr}$, $g^{ef}g_{gh}$, and $g\circ g$, which is of the form $N^4(t)f^2(t)q_{abcd}$. These products are resumed in Tables \ref{tab_kn_ricci} and \ref{tab_kn_scalar}, where we can see that they are smooth. In writing the tables, we made use of the following facts:
\begin{itemize}
	\item 
only some particular combinations of indices are allowed when contracting, and in the Kulkarni-Nomizu products,
	\item 
$g\circ g$ is of the form $N^4f^2q_{abcd}$, with $q_{abcd}$ smooth,
	\item 
$g_{ab}=N^2\tilde{g}_{ab}$, where $\tilde{g}_{ab}$ is smooth; $g^{ab}=N^{-2}f^{-2}\hat{g}^{ab}$, with $\hat{g}^{ab}$ smooth
	\item 
$g_{ij}=a^2h_{ij}=f^2N^2h_{ij}$; $g^{ij}=a^{-2}h^{ij}=f^{-2}N^{-2}h^{ij}$.
\end{itemize}

\begin{table}[!htbp]
\centering
\caption{Terms from $\Ric\circ g$.}
\label{tab_kn_ricci}
\begin{tabular}{|c|c|c|c|}
\toprule
Term &  \multicolumn{2}{c|}{Multiply with} & Term from $\Ric\circ g$ \\
\midrule
$u_{abcd}$ &  $g^{ef}$ & $g_{pq}$ & $u_{abcd} g^{ef} g_{pq}$ \\
{} &  ${}_{\{e,f\}\subset\{a,b,c,d\}}$ & ${}_{\{p,q\}\neq \{a,b,c,d\}-\{e,f\}}$ & {} \\\toprule
$g_{00,00}$ or $g_{00,\cocontr}g_{00,\cocontr}$ & $-N^{-2}$ & $a^2h_{ij}$ &  $-u_{0000}f^2h_{ij}$ \\\midrule
$g_{ij,kl}=a^2 \partial_k \partial_l h_{ij}$ & $a^{-2}h^{ef}$ & $g_{pq}$ & $h^{ef}\partial_k \partial_l h_{ij}g_{pq}$ \\\midrule
$g_{ij,\cocontr}g_{kl,\cocontr}=f^2 \tilde{u}_{ijkl}$ & $f^{-2}N^{-2}h^{ef}$ & $N^2\tilde{g}_{pq}$ & $\tilde{u}_{ijkl}h^{ef}\tilde{g}_{pq}$ \\\midrule
$g_{ij,00}$ or $g_{ij,\cocontr}g_{00,\cocontr}$ & $g^{i0}=0$ & $g_{pq}$ & $0$ \\\midrule
$g_{ij,00}$ or $g_{ij,\cocontr}g_{00,\cocontr}$ & $g^{00}=N^{-2}$ & $N^2\tilde{g}_{pq}$ & $\tilde{g}_{pq}u_{abcd}$ \\\midrule
$g_{ij,00}$ or $g_{ij,\cocontr}g_{00,\cocontr}$ & $g^{ij}=a^{-2}h^{ij}$ & $a^2h_{pq}$ ($p,q\neq 0$) & $h^{ik}h_{pq}u_{abcd}$ \\\midrule
$g_{ij,k0}=a v_{ijk}$ & $g^{ef}=a^{-2}h^{ef}$ & $a^2h_{pq}$ ($p,q\neq 0$) & $a h^{ef}h_{pq} v_{ijk}$ \\\bottomrule
\end{tabular}
\end{table}

\begin{table}[!htbp]
\centering
\caption{Terms from $R g\circ g$.}
\label{tab_kn_scalar}
\begin{tabular}{|c|c|c|c|}
\toprule
Term &  \multicolumn{2}{c|}{Multiply with} & Term from $R(g\circ g)$ \\
\midrule
$u_{abcd}$ &  $g^{ef}g^{gh}$ & $(g\circ g)_{pqrs}$ & $u_{abcd} g^{ef}g^{gh} (g\circ g)_{pqrs}$ \\
{} &  ${}_{\{e,f,g,h\}=\{a,b,c,d\}}$ & {} & {} \\\toprule
$g_{00,00}$ or $g_{00,\cocontr}g_{00,\cocontr}$ & $-N^{-4}$ & $N^4f^2q_{abcd}$ &  $-u_{0000}f^2q_{abcd}$ \\\midrule
$g_{ij,kl}=a^2 \partial_k \partial_l h_{ij}$ & $a^{-4}h^{ef}h^{gh}$ & $N^4f^2q_{abcd}$ & $N^2q_{abcd} h^{ef}h^{gh}\partial_k \partial_l h_{ij}$ \\\midrule
$g_{ij,\cocontr}g_{kl,\cocontr}=f^2 \tilde{u}_{ijkl}$ & $f^{-4}N^{-4}h^{ef}h^{gh}$ & $N^4f^2q_{abcd}$ & $\tilde{u}_{ijkl}h^{ef}h^{gh}q_{abcd}$ \\\midrule
$g_{ij,00}$ or $g_{ij,\cocontr}g_{00,\cocontr}$ & $g^{00}g^{ij}=N^{-4}f^{-2}h^{ij}$ & $N^4f^2q_{abcd}$ & $h^{ij}u_{abcd}q_{abcd}$ \\\midrule
$g_{ij,k0}=a v_{ijk}$ & $0$ & $N^4f^2q_{abcd}$ & $0$ \\\bottomrule
\end{tabular}
\end{table}

This concludes the proof that all terms contained in $R_{abcd}$, $E_{abcd}$, and $S_{abcd}$ are smooth. Hence, the metric \eqref{eq_metric_a_Nh} is {\quasireg}.
\end{proof}

This model can be connected with a cosmological model proposed by P. Fiziev and D. V. Shirkov \cite{FS2012Axial}, which consists in a dimensional reduction of spacetime (in topology and geometry) at the Big-Bang. When $a(t)$ vanishes, the space can be considered, in a way, to shrink to a point, so that one remains only with the time dimension. There are several properties of our model which seem to support that, at least, the universe behaves ``as if'' the dimensionality is reduced: the reduction of the rank of the metric, of the dimension of the cotangent space at such singular points, and by a certain independence of fields with the degenerate directions. In fact, if $N(t)=0$ as well, it appears that the Universe emerged from a dimensionless point. Yet, the conditions in which such a dimensional reduction can be considered are not clear. Moreover, for our {\quasireg} singularities it seems to be, at least at this time, better to maintain four topological dimensions, because such a reduction may, in general, lose initial data.

\section{Conclusion}

The {\semireg} singularities are well-behaved from many viewpoints, allowing us to perform the most important operations which are allowed by {\semiriem} manifolds with regular metric tensor. When, in addition, they are {\quasireg}, we can write a smooth expanded version of Einstein's equation (\sref{s_einstein_exp_qreg}). 

The {\quasireg} singularities offer a nice surprise, since they have vanishing Weyl curvature $C_{abcd}$. It follows that any {\quasireg} Big-Bang singularity also satisfies the {\Wch} (\sref{s_wch_thm}).

As a main application, we studied in \sref{s_wch_ex} a cosmological model which extends {\FLRW} and Bianchy cosmological models, by dropping the isotropy and homogeneity conditions. This generality is more realistic from physical viewpoint, since our Universe appears homogeneous and isotropic only at very large scales. This model contains as particular cases, in addition to {\FLRW}, also the isotropic singularities.

\section*{Acknowledgements}
The author thanks P. Fiziev and D. V. Shirkov from Bogoliubov Laboratory of Theoretical Physics, JINR, Dubna, for helpful discussions and advice concerning the form and the content of this paper. The author thanks an anonymous referee for the valuable comments and suggestions to improve the clarity and the quality of this paper.

\pagebreak
\bibliographystyle{unsrt}

\end{document}